\documentclass[preprint]{imsart}
\usepackage{booktabs} 
\usepackage{amsthm}
\usepackage{amsmath}
\usepackage{amssymb}
\usepackage{amstext} 
\usepackage{array}   
\usepackage{enumitem} 
\setlist[itemize]{noitemsep} 
\allowdisplaybreaks
\usepackage{multirow}
\usepackage[flushleft]{threeparttable}
\usepackage{verbatim}
\usepackage{bbm}
\usepackage{longtable}
\usepackage{rotating}
\usepackage{caption}
\usepackage{graphicx}
\usepackage{subcaption}
\usepackage{mwe}
\usepackage{bm}
\usepackage{indentfirst}
\usepackage{verbatim}
\usepackage{forest}

\usepackage{natbib}
\usepackage[colorlinks,citecolor=black,urlcolor=black,filecolor=black,backref=page]{hyperref}

\newcommand{\R}{\mathbb{R}}

\newtheorem{theorem}{Theorem}
\newtheorem{lemma}{Lemma}

\begin{document}

\begin{frontmatter}

\title{Bayesian Change Point Detection for Functional Data}

\author{Xiuqi Li}
\thanks{Operations Research Graduate Program, North Carolina State University, xli35@ncsu.edu}
\and
\author{Subhashis Ghosal}
\thanks{Department of Statistics, North Carolina State University, sghosal@stat.ncsu.edu}

\begin{abstract}
We propose a Bayesian method to detect change points for functional data. We extract the features of a sequence of functional data by the discrete wavelet transform (DWT), and treat each sequence of feature independently. We believe there is potentially a change in each feature at possibly different time points. The functional data evolves through such changes throughout the sequences of observations. The change point for this sequence of functional data is the cumulative effect of changes in all features. We assign the features with priors which incorporate the characteristic of the wavelet coefficients. Then we compute the posterior distribution of change point for each sequence of feature, and define a matrix where each entry is a measure of similarity between two functional data in this sequence. We compute the ratio of the mean similarity between groups and within groups for all possible partitions, and the change point is where the ratio reaches the minimum. We demonstrate this method using a dataset on climate change.

\end{abstract}
\end{frontmatter}

\section{Introduction}
\par Change point detection has always been an important aspect in data analysis. In recent years, there are increasing interests in developing methods to detect change point for functional data. \cite{Berkes} developed a method that works with the difference of mean functions projected on the principal components of the data. \cite{Zhang} developed a self-normalization (SN) based test to identify potential change points in the dependence structure of functional observations. \cite{Aston} also proposed a method to detect change points for dependent functional observations, and they were particularly interested in the case where the change point is an epidemic change (a change occurs and then the observations return to baseline at a later time). \cite{Sharipov} developed a new test for structural changes in functional data that based on Hilbert space theory and critical values are deduced from bootstrap iterations.  \cite{Aue2018} proposed a method to uncover structural breaks in functional data that does not rely on dimension reduction techniques.
\par In this paper, we propose a Bayesian method to detect change points for functional data. We extract the features of a sequence of functional data by the discrete wavelet transform (DWT), and treat each sequence of feature independently. We believe there is potentially a change in each feature at possibly different time points. The functional data evolves through such changes throughout the sequence of observations. The change point for this sequence of functional data is the cumulative effect of changes in all features. Such gradual evolutionary model for changes seems to be very appropriate for functional data, since functions have many aspects and it is hard to imagine that all those change at once. When such a cumulative effect becomes substantial to make the following functional observations significantly dissimilar with the previous ones---in that the variation across the two groups relative to the variation within the two groups is the maximum at that time point. We assign the features with priors which incorporate the characteristic of the wavelet coefficients. Then we compute the posterior distribution of change point for each sequence of feature, and define a matrix where each entry is a measure of similarity between two functional data in this sequence. We compute the ratio of the mean similarities between groups and within groups for all possible partitions, and the change point is where the ratio reaches the minimum. Once we have detected one change point, we can successively apply the procedure to subgroups divided by the change point. We can continue finding the change points in subgroups until a stopping criterion has been met. For example, we can stop if a certain number of change points have been detected, or there is no significant difference in the previous and following observations. Thus, this method can be inherently extended to multiple change points detection.

\section{Model}
\par We follow the formulation of \cite{Suarez} for the structure of functional observations, who applied their model in the context of clustering. We extend their approach to change point detection for functional data, which can be regarded as a special case of clustering with the constraint that for each characteristic, there are at most two clusters and they are linearly ordered. Suppose that the functional observations arise from true signals $f_i(t)$, $t\in[0,1]$, $i=1,\dots,n$, corrupted by some noise process, where $n$ denotes the sample size. We observe the functional data at some discrete time points. Then the model can be represented as 
\begin{align}
Y_i(T_l)= f_i(T_l)+\varepsilon_{il}, 
\end{align}
where $\varepsilon_{il}$ is assumed to follow a normal distribution with mean 0 and variance $\sigma^2$, and is independent across $i$ and $l$. Let $Y_i=(Y_i(T_1),\dots,Y_i(T_m))^T$ be the $i$th obsvervation at points $T_1,\dots,T_m$, where $T_l\in[0,1]$, for $l=1,\dots,m$. Similarly, let $f_i=(f_i(T_1),\dots,f_i(T_m))^T$, and $\varepsilon_i=(\varepsilon_i(T_1),\dots,\varepsilon_i(T_m))^T$. For functional data, the discrete wavelet transform (DWT) is one of the most common feature extraction technique. To implement the DWT, $m$ needs to be a power of 2, and $T_1,\dots,T_m$ need to be equidistant. For $m$ that is not a power of 2, we can first  smooth to obtain a function, and then take a power of 2 number of discrete points from that function. In terms of the orthonormal basis $\{\phi_0\}\cup\{\psi_{jk}:j=0,\dots,J-1,k=0,\dots,2^{j}-1\}$, we can define the following DWT operator \citep{Antoniadis}:
\begin{align}
W: \R^m\rightarrow\R^m, \:\: f\rightarrow(\alpha_0,\beta_{0},\dots,\beta_{J-1}),
\end{align}
with $\beta_j=(\beta_{j,0},\dots,\beta_{j,2^{j}-1})$.
Applying the DWT operator on $Y_i=f_i+\varepsilon_i$, then we have
\begin{align}
WY_i=Wf_i+W\varepsilon_i,
\end{align}
where $W\varepsilon_i\stackrel{d}{=}\varepsilon_i$ by the orthogonality of $W$. Let $\alpha_0$ denote the scaling coefficient at the level 0, and $\beta_{jk}$ be the wavelet coefficients at the multiresolution level $j,k$. As a result, (3) can be rewritten as
\begin{align}
a^{(i)}_0=\alpha^{(i)}_0+e^{(i)}_0, \:\: b^{(i)}_{jk}=\beta^{(i)}_{jk}+e^{(i)}_{jk},	
\end{align}
where $e_0^{(i)}$ and $e_{jk}^{(i)}$ follow a normal distribution with mean 0 and variance $\sigma^2$, for $k=0,\dots,2^{j}-1, j=0,\dots,J-1$.
\par When the functional data are (essentially) observed continuously in time, we also consider the following infinite Gaussian white noise model
\begin{align}
dY_i(t)=f_i(t)dt+\sigma dB_i(t),
\end{align}
where $B_i(\cdot)$ are independent Brownian motions on $[0,1]$. Let
\begin{align}
\label{eqn:eqlabel}
\begin{split}
&a_0^{(i)}=\int_{0}^{1}\phi_0(t)dY_i(t), \: \alpha_{0}^{(i)}=\int_{0}^{1}\phi_0(t)f_i(t)dt,\\
&b_{jk}^{(i)}=\int_{0}^{1}\psi_{jk}(t)dY_i(t), \: \beta_{jk}^{(i)}=\int_{0}^{1}\psi_{jk}(t)f_i(t)dt,\\
&e_0^{(i)}=\sigma\int_{0}^{1}\phi_0(t)dB_i(t), \: e_{jk}^{(i)}=\sigma\int_{0}^{1}\psi_{jk}(t)dB_i(t).
\end{split}
\end{align}
Then $e_0^{(i)}$ and $e_{jk}^{(i)}$ follow the normal distribution with mean 0 and variance $\sigma^2$, for $k=0,\dots,2^j-1$, $j=1,2,\dots$, independent of each other, for each $i=1,\dots,n$. 
\par To detect the change point of this sequence of functional data, we first find the change in each component, that is, we detect the change for each feature $\beta_{jk}$, and we decide the overall change point from them. 
\par In the section on posterior consistency, we state the results only in terms of the infinite model. However, in practice, we can only work with the finite model. By letting $\beta_{jk}^{(i)}=0$ for all $j>J$, the infinite model can be related to the finite one with a random $J$. If the coefficients are obtained following the schema of \cite{Abramovich1998}, then $J$ will have a limiting Poisson distribution by Proposition 1 of \cite{Suarez}. Under this schema, the total number of nonzero coefficients also has a limiting Poisson distribution.

\section{Prior Distributions}
\par For each $\beta^{(i)}_{jk}$, we define the following probabilities:
\begin{align}
\label{eqn:eqlabel}
\begin{split}
&\text{P}(\beta^{(i)}_{jk}\neq 0)=\pi_{j}, \ \text{P}(\beta^{(i)}_{jk}=0)=1-\pi_{j}.
\end{split}
\end{align}
\par As the wavelet coefficients of a signal function are sparse, \cite{Abramovich1998} proposed the following priors incorporating this characteristic feature of wavelet coefficients:
\begin{align}
\label{eqn:eqlabel}
\begin{split}
\beta^{(i)}_{jk}\stackrel{\text{ind}}{\sim} \pi_{j}\text{N}(0,c^2_{j}\sigma^2)+(1-\pi_{j})\delta_0,
\end{split}
\end{align}
where $\delta_0$ is a point mass at 0, and the hyperparameters in (8) are given by 
\begin{align}
c_j^2=\nu_12^{-\gamma_1j},\: \pi_j=\text{min}(1,\nu_22^{-\gamma_2j}),\: j=0,\dots,J-1,
\end{align}
and $\nu_1, \nu_2, \gamma_1\geq0$, and $0\leq\gamma_2\leq1$.  A vague prior is placed on $\alpha_0$.  
\par Let $\psi$ be a mother wavelet function of regularity $r$. Consider constants $s$, $p$ and $q$ such that 
$\text{max}(0,1/p - 1/2) < s < r$, $1\leq p,q\leq\infty$. If either
\begin{align}
\label{eqn:eqlabel}
\begin{split}
&s+\frac{1}{2}-\frac{\gamma_2}{p}-\frac{\gamma_1}{2}<0,
\end{split}
\end{align}
or
\begin{align}
\label{eqn:eqlabel}
\begin{split}
&s+\frac{1}{2}-\frac{\gamma_2}{p}-\frac{\gamma_1}{2}=0, \:\text{and}\:0\leq\gamma_2<1, 1\leq p<\infty, q=\infty,
\end{split}
\end{align}
then $f\in\mathcal{B}^s_{p,q}$ almost surely, where $\mathcal{B}^s_{p,q}$ denotes Besov space of index $(p,q)$ and smoothness $s$ \citep{Abramovich1998}.
\par The prior on $\sigma$ is given by
\begin{align}
\sigma^2\sim \text{IG}(\theta,\lambda),
\end{align}
where IG stands for the inverse gamma distribution. Let $g$ denote the density funciton of the inverse-gamma distribution.
\section{Posterior Probabilities of Change Point}
\par For any $j,k$, let $\tau_{jk}$ denote the change point, and let $\tau_{jk}$ take possible values $1,\dots,n$. Let $\rho_i$ denote the prior probability of changing at point $i$, where $\rho_i>0$, $\sum^{n}_{i=1}\rho_i=1$, and  $i=1,\dots,n$. Then the posterior probability of $\tau_{jk}=i$ is
\begin{align}
\text{P}(\tau_{jk}=i|b^{(1)}_{jk},\dots,b^{(n)}_{jk})=\frac{\text{P}(b^{(1)}_{jk},\dots,b^{(n)}_{jk}|\tau_{jk}=i)\rho_i}{\sum^{N}_{l=1}\text{P}(b^{(1)}_{jk},\dots,b^{(n)}_{jk}|\tau_{jk}=l)\rho_l}.
\end{align}
\par The main problem is to compute the marginal likelihood $\text{P}(b^{(1)}_{jk},\dots,b^{(n)}_{jk}|\tau_{jk}=i)$. When $\tau_{jk}=1$, it is the initial state meaning no change. For $\tau_{jk}=2,\dots,n$, the marginal likelihood is derived from four scenarios: change from zero to zero (which is no change), change from zero to non-zero, change from non-zero to zero, and change from non-zero to non-zero. 

\subsection{Initial State}
\par When $\tau_{jk}=1$, this is the initial state. 
If the initial state is zero, then the marginal likelihood is given by 
\begin{align}
(1-\pi_j)\int \Bigl\{\prod^{n}_{i=1}\phi(b^{(i)}_{jk};0,\sigma^2)\Bigr\} g(\sigma^2;\theta,\lambda)d\sigma.
\end{align}
If the initial state is non-zero, then we have 
\begin{align}
\pi_j\int\int \Bigl\{\prod^{n}_{i=1}\phi(b^{(i)}_{jk};\xi,\sigma^2)\Bigr\} \phi(\xi;0,c^2_j\sigma^2) g(\sigma^2;\theta,\lambda)d\xi d\sigma.
\end{align}
Thus, the marginal likelihood of the initial state is 
\begin{align}
\label{eqn:eqlabel}
\begin{split}
&\text{P}(b^{(1)}_{jk},\dots,b^{(n)}_{jk}|\tau_{jk}=1)\\
&\phantom{\text{P}}=(1-\pi_j)(2\pi)^{-n/2}\frac{\lambda^\theta}{\Gamma(\theta)}\frac{\Gamma(n/2+\theta)}{\displaystyle [\frac{\sum^{n}_{i=1}(a^{(i)}_0)^2}{2}+\lambda]^{(n/2+\theta)}}\\
&\phantom{=(1}+\pi_j(2\pi)^{-n/2}\frac{\lambda^\theta}{\Gamma(\theta)}(c^2_jn+1)^{-1/2}\\
&\phantom{+\pi_j(2\pi)}\times\frac{\Gamma(n/2+\theta)}{\displaystyle \Big[\frac{\sum^{n}_{i=1}(a^{(i)}_0)^2}{2}-\frac{c^2_j}{c^2_jn+1}\frac{(\sum^{n}_{i=1}a^{(i)}_0)^2}{2}+\lambda\Big]^{(n/2+\theta)}}.
\end{split}
\end{align}

\subsection{Non-initial State}
For $\tau_{jk}=i$, where $i=2,\dots,n$, if changing from zero to zero, that is, no change, then the marginal likelihood is
\begin{align}
(1-\pi_j)^2\int \Bigl\{\prod^{n}_{i=1}\phi(b^{(i)}_{jk};0,\sigma^2)\Bigr\} g(\sigma^2;\theta,\lambda)d\sigma.
\end{align}
If changing from zero to non-zero at $\tau_{jk}=i$, then the marginal likelihood is 
\begin{align}
\label{eqn:eqlabel}
\begin{split}
(1-\pi_j)\pi_j\int\int& \Bigl\{\prod^{i-1}_{l=1}\phi(b^{(l)}_{jk};0,\sigma^2)\Bigr\} \Bigl\{\prod^{n}_{l=i}\phi(b^{(l)}_{jk};\xi,\sigma^2)\Bigr\} \phi(\xi;0,c^2_j\sigma^2) g(\sigma^2;\theta,\lambda)d\xi d\sigma.
\end{split}
\end{align}
If changing from non-zero to zero at $\tau_{jk}=i$, then the marginal likelihood is
\begin{align}
\label{eqn:eqlabel}
\begin{split}
\pi_j(1-\pi_j)\int\int& \Bigl\{\prod^{i-1}_{l=1}\phi(b^{(l)}_{jk};\xi,\sigma^2)\Bigr\} \Bigl\{\prod^{n}_{l=i}\phi(b^{(l)}_{jk};0,\sigma^2)\Bigr\}\phi(\xi;0,c^2_j\sigma^2) g(\sigma^2;\theta,\lambda)d\xi d\sigma.
\end{split}
\end{align}
If changing from non-zero to non-zero at $\tau_{jk}=i$, then the marginal likelihood is 
\begin{align}
\label{eqn:eqlabel}
\begin{split}
\pi_j^2\int\int\int& \Bigl\{\prod^{i-1}_{l=1}\phi(b^{(l)}_{jk};\xi_1,\sigma^2)\Bigr\} \Bigl\{\prod^{n}_{l=i}\phi(b^{(l)}_{jk};\xi_2,\sigma^2)\Bigr\}\\
& \phi(\xi_1;0,c^2_j\sigma^2)\phi(\xi_2;0,c^2_j\sigma^2) g(\sigma^2;\theta,\lambda)d\xi_1 d\xi_2 d\sigma.
\end{split}
\end{align}
Thus, we have
\begin{align}
\label{eqn:eqlabel}
\begin{split}
&\text{P}(b^{(1)}_{jk},\dots,b^{(n)}_{jk}|\tau_{jk}=i)\\
&\phantom{\text{P}}=(1-\pi_j)^2(2\pi)^{-n/2}\frac{\lambda^\theta}{\Gamma(\theta)}\frac{\Gamma(n/2+\theta)}{\displaystyle\Big[\frac{\sum^{n}_{l=1}(b^{(l)}_{jk})^2}{2}+\lambda\Big]^{(n/2+\theta)}}\\	
&\phantom{=(1}+\pi_j(1-\pi_j)(2\pi)^{-n/2}\frac{\lambda^\theta}{\Gamma(\theta)}[c^2_j(i-1)+1]^{-1/2}\frac{\Gamma(n/2+\theta)}{\displaystyle[B_{i-1}+\lambda]^{(n/2+\theta)}}\\
&\phantom{=(1}+(1-\pi_j)\pi_j(2\pi)^{-n/2}\frac{\lambda^\theta}{\Gamma(\theta)}[c^2_j(n-i+1)+1]^{-1/2}\frac{\Gamma(n/2+\theta)}{\displaystyle[B_{n-i+1}+\lambda]^{(n/2+\theta)}}\\
&\phantom{=(1}+\pi_j^2(2\pi)^{-n/2}\frac{\lambda^\theta}{\Gamma(\theta)}[c^2_j(i-1)+1]^{-1/2}[c^2_j(n-i+1)+1]^{-1/2}\\
&\phantom{+\pi_j^2(2\pi)}\times\frac{\Gamma(n/2+\theta)}{\displaystyle\Big[B_{i-1}+B_{n-i+1}-\frac{\sum^{n}_{l=1}(b^{(l)}_{jk})^2}{2}+\lambda\Big]^{(n/2+\theta)}},
\end{split}
\end{align}
where
\begin{align}
B_{i-1}=\frac{\sum^{n}_{l=1}(b^{(l)}_{jk})^2}{2}-\frac{c^2_j}{c^2_j(i-1)+1}\frac{(\sum^{i-1}_{l=1}b^{(l)}_{jk})^2}{2},
\end{align}
and
\begin{align}
\tilde{B}_{n-i+1}=\frac{\sum^{n}_{l=1}(b^{(l)}_{jk})^2}{2}-\frac{c^2_j}{c^2_j(n-i+1)+1}\frac{(\sum^{n}_{l=i}b^{(l)}_{jk})^2}{2}.
\end{align}

\par Similarly, we can compute the marginal likelihood $\text{P}(a^{(1)}_{0},\dots,a^{(n)}_{0}|\tau_{0}=i)$, where $\tau_{0}$ denotes the change point in $a^{(1)}_{0},\dots,a^{(n)}_{0}$, and obtain the posterior probability of $\tau_{0}=i$ through Bayes's rule.

\section{Change Point Detection}
The change point of a sequence of functional data is the accumulative effect of all features where the contrast is the largest before and after. Since its the special case of clustering, following \cite{Suarez}, to quantify the similarity between two funcitonal data, we need to consider the following similarity matrix. Suppose that there are $J$ levels and $i<i'$. Then the similarity between $i$th and $i'$th functional data is 
\begin{align}
S(i,i')=\frac{1}{2^J}\big[\mathbbm{1}(\alpha^{(i)}_{0}=\alpha^{(i')}_{0})+\sum^{J-1}_{j=0}\sum^{2^j-1}_{k=0}\mathbbm{1}(\beta^{(i)}_{jk}=\beta^{(i')}_{jk})\big].
\end{align}
\par For any $k$ that divides the data into two groups, we compute the ratio of the mean similarity between group and the mean similarity within group. We denote the ratio by $C(k)$. The change point is where this ratio is the minimum. Here we assume that $3\leq k\leq n-1$, which means there are at least two data points in each group. Then
\begin{align}
\text{argmin}_k\:\:C(k)=\frac{\sum_{1\leq i\leq k-1,k\leq j\leq N}S_{ij}/[(k-1)(n-k+1)]}{\{\sum_{1\leq i \leq j \leq k-1}S_{ij}+\sum_{k\leq i \leq j \leq N}S_{ij}\}/[{{k-1}\choose{2}}+{{n-k+1}\choose{2}}]}.
\end{align}

\par Since we cannot obtain the true value of $\alpha_0$ and $\beta_{jk}$, we take posterior expectation of (24) given the data. Then we have
\begin{align}
\label{eqn:eqlabel}
\begin{split}
&\text{E}(S(i,i')|a^{(1)}_{0},\dots,a^{(n)}_{0},b^{(1)}_{jk},\dots,b^{(n)}_{jk})\\
&\phantom{\text{E}}=
\frac{1}{2^J}\big[\text{P}(\alpha^{(i)}_{0}=\alpha^{(i')}_{0}|a^{(1)}_{0},\dots,a^{(n)}_{0})+\sum^{J-1}_{j=0}\sum^{2^j-1}_{k=0}\text{P}(\beta^{(i)}_{jk}=\beta^{(i')}_{jk}|b^{(1)}_{jk},\dots,b^{(n)}_{jk})\big].
\end{split}
\end{align}
$\text{P}(\beta^{(i)}_{jk}=\beta^{(i')}_{jk}|b^{(1)}_{jk},\dots,b^{(n)}_{jk})$ can be obtained from the expression for the posterior probability of the change point:
\begin{align}
\label{eqn:eqlabel}
\begin{split}
&\text{P}(\beta^{(i)}_{jk}=\beta^{(i')}_{jk}|b^{(1)}_{jk},\dots,b^{(n)}_{jk})\\
&\phantom{\text{P}}=\text{P}(\tau_{jk}\leq i \: \text{or} \: \tau_{jk}\geq i'+1|b^{(1)}_{jk},\dots,b^{(n)}_{jk})\\
&\phantom{\text{P}}=\sum^{i}_{t=1}\text{P}(\tau_{jk}=t|b^{(1)}_{jk},\dots,b^{(n)}_{jk})+\sum^{n}_{t=i'+1}\text{P}(\tau_{jk}=t|b^{(1)}_{jk},\dots,b^{(n)}_{jk}),
\end{split}
\end{align}
where $\text{P}(\tau_{jk}=t|b^{(1)}_{jk},\dots,b^{(n)}_{jk})$ are obtained from (13). Similarly, we can obtain $\text{P}(\alpha^{(i)}_{0}=\alpha^{(i')}_{0}|a^{(1)}_{0},\dots,a^{(n)}_{0})$.

\section{Posterior Consistency}
\par In this section, we state a posterior consistency result for the infinite model. With some minor notational modification, the result also holds for the finite model with a fixed depth $J$.
\par We study consistency in our model when $\sigma^2\rightarrow0$. This is equivalent to averaging $r$ i.i.d. replications of the observations with $r\rightarrow\infty$, and replacing $\sigma^2$ by $\sigma^2/r$ with a known $\sigma^2$. To simiplify notation, we assume that $\alpha_0^{(i)}=0$ for $i=1,\dots,n$. Let $\bm{f}=(f_1,\dots,f_n)$. Then the square of the norm on $\bm{f}$ is defined by
\begin{align}
\|\bm{f}\|^2=\sum_{i=1}^{n}\|f_i\|_2^2=\sum_{i=1}^{n}\sum_{j=0}^{\infty}\sum_{k=0}^{2^j-1}|\beta_{jk}^{(i)}|^2.
\end{align}
We define the square of Sobolev norm on $\bm{f}$ as
\begin{align}
\|\bm{f}\|_{\mathcal{H}_{n}^s}=\sum_{i=1}^{n} \sum_{j=0}^{\infty}2^{2js}\|\beta_{j\cdot}^{(i)}\|_2^2.
\end{align}
We denote this space by $\mathcal{H}_{n}^s$, where $s$ is the number of weak derivatives of the function in $L_2([0,1])$. Let $D_r$ be the set of all observations.

\begin{theorem}
	Let $\gamma_1>2s+1$, and $\bm{f}_0\in \mathcal{H}_n^s$ be the vector of true functions. Then the posterior is consistent, i.e., for any $\epsilon>0$, $\Pi(\|\bm{f}-\bm{f}_0\|<\epsilon|D_r)\rightarrow1$ in probability as $r\rightarrow\infty$.
\end{theorem}

\begin{proof}
	Let $\Pi$ be a prior on $\mathcal{H}^s_n$. Schwartz's theorem \citep{Schwartz} gives the strong consistency of the posterior distribution under approximate condition. According to Example 6.20 in \cite{Ghosal}, if the Kullback–Leibler property holds for the prior, then the posterior distribution is consistent in the weak topology. Thus, for prior in space $\mathcal{H}^s_n$,  we need $\bm{f}_0$ in the Kullback-Leibler support of $\Pi$. The Kullback–Leibler divergence is defined as $\mathcal{K}(\bm{f}_{0},\bm{f})=\sum_{i=1}^{n}\int f_{i,0}\log(f_{i,0}/f_i)d\mu$, where $\mu$ is a dominating measure on the space of $\bm{f}$. In other words, we want $\Pi\left(\mathcal{K}(\bm{f}_0,\bm{f})<\epsilon\right)>0$ for all $\epsilon>0$. The prior setting in Section 3 reduces $\mathcal{K}(\bm{f}_0,\bm{f})$ to the Kullback-Leibler divergence between two Gaussian distributions that is the Kullback-Leibler divergence between $\beta_{jk}$ and $\beta_{jk,0}$, and thus $\Pi\left(\mathcal{K}(\bm{f}_0,\bm{f})<\epsilon\right)$ is bounded by 
	\begin{align}
		\label{eqn:eqlabel}
		\begin{split}
			\Pi\left(\sum_{i=1}^{n}\sum_{j=0}^{\infty}\sum_{k=0}^{2^j-1}|\beta_{jk}^{(i)}-\beta_{jk,0}^{(i)}|^2<\epsilon^2\right),
		\end{split}
	\end{align}
	where $\beta_{jk,0}^{(i)}$ are the wavelet coefficients of the true function $f_{i,0}$. 
	\par First, we consider a bounded subset $\mathcal{H}_n^s(B)=\{\bm{f}\in\mathcal{H}_n^s,\|\bm{f}\|_{\mathcal{H}^s_n}<B\}$ of the Sobolev space. The Lemma 1 and 2 of \cite{Lian} imply that (30) is positive. Thus, for any $B>0$, we have
	\begin{align}
		\Pi\left(\bm{f}\in \mathcal{H}_n^s(B):\|\bm{f}-\bm{f}_0\|>\epsilon|D_r\right)\rightarrow 0 \:\text{in probability}.
	\end{align}
	\par To complete the proof, we need to show that $\lim_{B\to\infty}\sup_{r>0}\text{E}_{\bm{f}_0}\Pi(\mathcal{H}_n^s(B)^c|D_r)\\=0$.
	By Markov's inequality, we have
	\begin{align}
		\Pi(\mathcal{H}_n^s(B)^c|D_r)\leq B^{-2}\left\{\sum_{i=1}^{n}\sum_{j=0}^{\infty}2^{2js}\sum_{k=0}^{2^j-1}\text{E}\left(|\beta_{jk}^{(i)}|^2\bigg|D_r\right)\right\}.
	\end{align}
	The expectation can be bounded by
	\begin{align}
		\label{eqn:eqlabel}
		\begin{split}
			\text{E}\left(|\beta_{jk}^{(i)}|^2\bigg|D_r\right)&=\sum_{t=1}^{n}\text{E}\left(|\beta_{jk}^{(i)}|^2\bigg|\tau_{jk}=t,D_r\right)\Pi(\tau_{jk}=t|D_r)\\
			&\leq \max_{1\leq t\leq n}\text{E}\left(|\beta_{jk}^{(i)}|^2\bigg|\tau_{jk}=t,D_r\right).
		\end{split}
	\end{align}
	For $\tau_{jk}=t$, the posterior distribution of the common value $\xi$ of $\{\beta_{jk}^{(i)}:i=1,\dots,t-1\}$ given $(\beta_{jk}\neq 0,D_r)$ is proportional to
	\begin{align}
		\label{eqn:eqlabel}
		\begin{split}
			\Big\{&\prod_{l=1}^{t-1}\exp\{-\frac{1}{2\sigma^2/r}(b_{jk}^{(l)}-\xi)^2\}\Big\}\times\exp\{-\frac{1}{2c_j^2\sigma^2}\xi^2\}\\
			&\propto\exp\left\{-\frac{(t-1)c_j^2+1/r}{2c_j^2\sigma^2/r}\Big[\xi^2-2\frac{c_j^2\sum_{l=1}^{t-1}b_{jk}^{(l)}}{(t-1)c_j^2+1/r}\xi\Big]\right\},
		\end{split}
	\end{align}
	and hence the corresponding distribution is $\text{N}\left(\frac{c_j^2\sum_{l=1}^{t-1}b_{jk}^{(l)}}{(t-1)c_j^2+1/r},\frac{c_j^2\sigma^2}{1+(t-1)c_j^2r}\right)$. Thus, for $i<t$, we have
	\begin{align}
		\text{E}\left(|\beta_{jk}^{(i)}|^2\bigg|\tau_{jk}=t,D_r\right)=\frac{c_j^2\sigma^2}{1+(t-1)c_j^2r}+\left(\frac{c_j^2}{(t-1)c_j^2+1/r}\right)^2\left(\sum_{l=1}^{t-1}b_{jk}^{(l)}\right)^2.
	\end{align} 
	Similarly, for $i\geq t$,
	\begin{align}
	\label{eqn:eqlabel}
	\begin{split}
		\text{E}\left(|\beta_{jk}^{(i)}|^2\bigg|\tau_{jk}=t,D_r\right)=&\frac{c_j^2\sigma^2}{1+(n-t+1)c_j^2r}\\
		&+\left(\frac{c_j^2}{(n-t+1)c_j^2+1/r}\right)^2\left(\sum_{l=t}^{n}b_{jk}^{(l)}\right)^2.
	\end{split}
	\end{align} 
	Note that if $t=1$, we only need to consider (36). Both (35) and (36) can be bounded by 
	\begin{align}
		\text{E}\left(|\beta_{jk}^{(i)}|^2\bigg|\tau_{jk}=t,D_r\right)\leq c_j^2\sigma^2+\left(\frac{c_j^2}{c_j^2+1/r}\right)^2\left(\sum_{l=1}^{n}b_{jk}^{(l)}\right)^2.
	\end{align}
	Thus we have
	\begin{align}
		\label{eqn:eqlabel}
		\begin{split}
			\Pi(\mathcal{H}_n^s(B)^c|D_r)\leq B^{-2}\left\{\sum_{i=1}^{n}\sum_{j=0}^{\infty}2^{2js}\sum_{k=0}^{2^j-1}\Big[c_j^2\sigma^2+\Big(\frac{c_j^2}{c_j^2+1/r}\Big)^2\Big(\sum_{l=1}^{n}b_{jk}^{(l)}\Big)^2 \vphantom{\sum_{i=1}^{n}\frac{c_{00}^2\sigma^2}{1+c_{00}^2r}}\Big]\right\}.
		\end{split}
	\end{align}
	Now we take the expectation of (38) with respect to $\bm{f}_0$ to obtain
	\begin{align}
		\label{eqn:eqlabel}
		\begin{split}
			&\text{E}_{\bm{f}_0}\Pi(\mathcal{H}_n^s(B)^c|D_r)\\
			&\phantom{\text{E}}\leq B^{-2}\left\{n\sum_{j=0}^{\infty}2^{2js}\sum_{k=0}^{2^j-1}\Big[c_j^2\sigma^2+\Big(\frac{c_j^2}{c_j^2+1/r}\Big)^2\{\frac{n\sigma^2}{r}+\Big(\sum_{l=1}^{n}\beta_{jk,0}^{(l)}\Big)^2\}\Big] \right\}\\
			&\phantom{\text{E}}\leq B^{-2}\left\{n\sum_{j=0}^{\infty}2^{2js}\sum_{k=0}^{2^j-1}\Big[c_j^2\sigma^2+nc_j^2\sigma^2+\Big(\sum_{l=1}^{n}\beta_{jk,0}^{(l)}\Big)^2\Big] \right\}.
		\end{split}
	\end{align}
	Replacing the hyperparameters using (9), we can further bound (39) by
	\begin{align}
		\label{eqn:eqlabel}
		\begin{split}
			&B^{-2}\left\{n(n+1)\sigma^2\nu_1\sum_{j=0}^{\infty}2^{(2s+1-\gamma_1)j}+n^2\sum_{j=0}^{\infty}2^{2js}\sum_{k=0}^{2^j-1}\sum_{l=1}^{n}|\beta_{jk,0}^{(l)}|^2\right\}\\
			&\phantom{B}= B^{-2}\left\{n(n+1)\sigma^2\nu_1\sum_{j=0}^{\infty}2^{(2s+1-\gamma_1)j}+n^2\|\bm{f}_0\|^2_{\mathcal{H}^s_n}\right\}.
		\end{split}
	\end{align}
	Under the assumption that $\gamma_1>2s+1$, we have 
	\begin{align}
		n(n+1)\sigma^2\nu_1\sum_{j=0}^{\infty}2^{(2s+1-\gamma_1)j}+ n^2\|\bm{f}_0\|^2_{\mathcal{H}^s_n}<\infty.
	\end{align}
	Thus, (39) goes to 0 as $B\rightarrow\infty$.
\end{proof}

\par We also need to show that we find the right model. For given $j,k$, we can define the following structures 1 to 5 denoted by $S_{jk}^{1},\dots,S_{jk}^{5}$:
\begin{enumerate}
	\item Change from nonzero to nonzero at $\tau_{jk}=t$:\\ $S_{jk}^1=\{\tau_{jk}=t,\beta_{jk}^{(1)}=\dots=\beta_{jk}^{(t-1)}=\xi_1,\beta_{jk}^{(t)}=\dots=\beta_{jk}^{(n)}=\xi_2,\xi_1\neq\xi_2\}$;
	\item Change from nonzero to zero at $\tau_{jk}=t$:\\
	$S_{jk}^2=\{\tau_{jk}=t,\beta_{jk}^{(1)}=\dots=\beta_{jk}^{(t-1)}=\xi,\beta_{jk}^{(t)}=\dots=\beta_{jk}^{(n)}=0,\xi\neq0\}$;
	\item Change from zero to nonzero at $\tau_{jk}=t$:\\
	$S_{jk}^3=\{\tau_{jk}=t,\beta_{jk}^{(1)}=\dots=\beta_{jk}^{(t-1)}=0,\beta_{jk}^{(t)}=\dots=\beta_{jk}^{(n)}=\xi,\xi\neq0\}$;
	\item No change and the value is nonzero:\\
	$S_{jk}^4=\{\beta_{jk}^{(1)}=\dots=\beta_{jk}^{(n)}=\xi,\xi\neq0\}$;
	\item No change and the value is zero:\\
	$S_{jk}^5=\{\beta_{jk}^{(1)}=\dots=\beta_{jk}^{(n)}=0\}$.
\end{enumerate}
We define a compatible model as the structure that not only has the same change point as the true model, but also can have $\beta_{jk}^{(i)}$ values in the neighborhood of the true value $\beta_{jk,0}^{(i)}$. For example, if the true Structure is 5, then the compatible model can be Structure 1--5, because a nonzero value $\xi$ can be small enough to be in the neighborhood of 0. If the true Structure is 1, then the only compatible model is itself, because 0 cannot be in the neighborhood of a predetermined nonzero value. Table 1 shows the compatible models for each true structure. Theorem 1 implies that the posterior probability of $\beta_{jk}^{(i)}$ in any neighborhood of the true value $\beta_{jk,0}^{(i)}$ tends to 1. This shows that the posterior probability of all non-compatible models together tends to 0. Hence for consistency of model selection, we only need to consider compatible models. 
\begin{table}
	\caption{Compatible models}
	\centering
	\begin{tabular}{cc}
		\hline
	True Structure & Compatible Model \\ 
	\hline
	Structure 1 & Structure 1\\ 
	Structure 2 & Structure 1,2\\
	Structure 3 & Structure 1,3\\
	Structure 4 & Structure 1,4\\
	Structure 5 & Structure 1,2,3,4,5\\
	\hline
\end{tabular}
\end{table}

\begin{lemma}
	Let $S_{jk,0}$ denote the true structure for given $j,k$. Then $\Pi(S_{jk}=S_{jk,0}|D_r)\rightarrow1$ in probability as $r\rightarrow\infty$.
\end{lemma}

\begin{proof}
	It suffices to show that the ratio of the marginal likelihood of a compatible structure other than the true structure and the true structure goes to zero in probability. In this proof, we only show the cases when the true parameter has Structures 4 or 5. The proofs for other cases follow from similar arguments.
	\par First, we need to compute the following marginal likelihoods with a known $\sigma^2$. The marginal likelihood for Structure 1 is 
	\begin{align}
		\label{eqn:eqlabel}
		\begin{split}
			&\text{P}(b_{jk}^{(1)},\dots,b_{jk}^{(n)}|S_{jk}^1)\\
			&\phantom{\text{P}}=\int\int \Bigl\{\prod^{t-1}_{l=1}\phi(b^{(l)}_{jk};\xi_1,\sigma^2/r)\Bigr\} \Bigl\{\prod^{n}_{l=t}\phi(b^{(l)}_{jk};\xi_2,\sigma^2/r)\Bigr\}\\ &\phantom{\int\int}\times\phi(\xi_1;0,c^2_j\sigma^2)\phi(\xi_2;0,c^2_j\sigma^2)d\xi_1 d\xi_2\\
			&\phantom{\text{P}}=\big(c_j^2r(t-1)+1\big)^{-1/2}\big(c_j^2r(n-t+1)+1\big)^{-1/2}(2\pi\sigma^2/r)^{-n/2}\\
			&\phantom{=\big(c_j^2}\times\exp\Big\{\frac{r}{2\sigma^2}\Big[\frac{c_j^2}{c_j^2(t-1)+1/r}\big(\sum_{l=1}^{t-1}b_{jk}^{(l)}\big)^2 \\
			&\phantom{times\exp\Big\{\frac{r}{2\sigma^2}}+\frac{c_j^2}{c_j^2(n-t+1)+1/r}\big(\sum_{l=t}^{n}b_{jk}^{(l)}\big)^2-\sum_{l=1}^{n}(b_{jk}^{(l)})^2\Big]\Big\}.
		\end{split}
	\end{align}
	The marginal likelihood for Structure 2 is
	\begin{align}
		\label{eqn:eqlabel}
		\begin{split}
			&\text{P}(b_{jk}^{(1)},\dots,b_{jk}^{(n)}|S_{jk}^2)\\
			&\phantom{\text{P}}=\int \Bigl\{\prod^{t-1}_{l=1}\phi(b^{(l)}_{jk};\xi,\sigma^2/r)\Bigr\} \Bigl\{\prod^{n}_{l=t}\phi(b^{(l)}_{jk};0,\sigma^2/r)\Bigr\} \phi(\xi;0,c^2_j\sigma^2)d\xi\\
			&\phantom{\text{P}}=\big(c_j^2r(t-1)+1\big)^{-1/2}(2\pi\sigma^2/r)^{-n/2}\\
			&\phantom{=\big(c_j^2}\times\exp\left\{\frac{r}{2\sigma^2}\left[\frac{c_j^2}{c_j^2(t-1)+1/r}\left(\sum_{l=1}^{t-1}b_{jk}^{(l)}\right)^2-\sum_{l=1}^{n}(b_{jk}^{(l)})^2\right]\right\}.
		\end{split}
	\end{align}
	The marginal likelihood for Structure 3 is
	\begin{align}
		\label{eqn:eqlabel}
		\begin{split}
			&\text{P}(b_{jk}^{(1)},\dots,b_{jk}^{(n)}|S_{jk}^3)\\
			&\phantom{\text{P}}=\int \Bigl\{\prod^{t-1}_{l=1}\phi(b^{(l)}_{jk};0,\sigma^2/r)\Bigr\} \Bigl\{\prod^{n}_{l=t}\phi(b^{(l)}_{jk};\xi,\sigma^2/r)\Bigr\} \phi(\xi;0,c^2_j\sigma^2)d\xi\\
			&\phantom{\text{P}}=\big(c_j^2r(n-t+1)+1\big)^{-1/2}(2\pi\sigma^2/r)^{-n/2}\\
			&\phantom{\big(c_j^2}\times\exp\left\{\frac{r}{2\sigma^2}\left[\frac{c_j^2}{c_j^2(n-t+1)+1/r}\left(\sum_{l=t}^{n}b_{jk}^{(l)}\right)^2-\sum_{l=1}^{n}(b_{jk}^{(l)})^2\right]\right\}.
		\end{split}
	\end{align}
	The marginal likelihood for Structure 4 is
	\begin{align}
		\label{eqn:eqlabel}
		\begin{split}
			&\text{P}(b_{jk}^{(1)},\dots,b_{jk}^{(n)}|S_{jk}^4)\\
			&\phantom{\text{P}}=\int \Bigl\{\prod^{n}_{l=1}\phi(b^{(l)}_{jk};\xi,\sigma^2/r)\Bigr\}\phi(\xi;0,c^2_j\sigma^2)d\xi\\
			&\phantom{\text{P}}=\big(c_j^2rn+1\big)^{-1/2}(2\pi\sigma^2/r)^{-n/2}\\
			&\phantom{=\big(c_j^2}\times\exp\left\{\frac{r}{2\sigma^2}\left[\frac{c_j^2}{c_j^2n+1/r}\left(\sum_{l=1}^{n}b_{jk}^{(l)}\right)^2-\sum_{l=1}^{n}(b_{jk}^{(l)})^2\right]\right\}.
		\end{split}
	\end{align}
	The marginal likelihood for Structure 5 is
	\begin{align}
		\label{eqn:eqlabel}
		\begin{split}
			&\text{P}(b_{jk}^{(1)},\dots,b_{jk}^{(n)}|S_{jk}^5)\\
			&\phantom{\text{P}}=\prod^{n}_{l=1}\phi(b^{(l)}_{jk};0,\sigma^2/r)\\
			&\phantom{\text{P}}=(2\pi\sigma^2/r)^{-n/2}\exp\left\{-\frac{r}{2\sigma^2}\sum_{l=1}^{n}(b_jk^{(l)})^2\right\}.
		\end{split}
	\end{align}
	\par If the true parameter has Structure 4, and the compatible model is Structure 1, then we have the following marginal likelihood ratio: 
	\begin{align}
		\label{eqn:eqlabel}
		\begin{split}
			&\frac{\text{P}(b_{jk}^{(1)},\dots,b_{jk}^{(n)}|S_{jk}^1)}{\text{P}(b_{jk}^{(1)},\dots,b_{jk}^{(n)}|S_{jk}^4)}\\
			&=\Big[\frac{c_j^2rn+1}{\left(c_j^2r(t-1)+1\right)\left(c_j^2r(n-t+1)+1\right)}\Big]^{1/2}\\
			&\phantom{=\Big[}\times\exp\Big\{\frac{r}{2\sigma^2}\Big[\frac{c_j^2}{c_j^2(t-1)+1/r}\Big(\sum_{l=1}^{t-1}b_{jk}^{(l)}\Big)^2+\frac{c_j^2}{c_j^2(n-t+1)+1/r}\Big(\sum_{l=t}^{n}b_{jk}^{(l)}\Big)^2\\ &\phantom{times\exp\Big\{}-\frac{c_j^2}{c_j^2n+1/r}\Big(\sum_{l=1}^{n}b_{jk}^{(l)}\Big)^2\Big]\Big\}.
		\end{split}
	\end{align}
	The first term in the squre root goes to 0 as $r\rightarrow\infty$. Hence it suffices to show that the form inside the exponential is $\mathcal{O}_p(1)$. Being a special case of clustering, our situation is similar to that of \cite{Suarez}, but it seems that their argument is incomplete as they overlooked a factor $r$. For the sake of completeness, we present the argument, which can also be used to complete the proof Lemma 1 of \cite{Suarez}.
	\par As $r\rightarrow\infty$, we have
	\begin{align}
		\label{eqn:eqlabel}
		\begin{split}
			&\frac{r}{2\sigma^2}\Big[\frac{c_j^2}{c_j^2(t-1)+1/r}\Big(\sum_{l=1}^{t-1}b_{jk}^{(l)}\Big)^2+\frac{c_j^2}{c_j^2(n-t+1)+1/r}\Big(\sum_{l=t}^{n}b_{jk}^{(l)}\Big)^2\\ &\phantom{\frac{r}{2\sigma^2}}-\frac{c_j^2}{c_j^2n+1/r}\Big(\sum_{l=1}^{n}b_{jk}^{(l)}\Big)^2\Big]\\
			&\phantom{-}=\frac{r}{2\sigma^2}\Big[\frac{\Big(\sum_{l=1}^{t-1}b_{jk}^{(l)}\Big)^2}{(t-1)+\mathcal{O}(1/r)}+\frac{\Big(\sum_{l=t}^{n}b_{jk}^{(l)}\Big)^2}{(n-t+1)+\mathcal{O}(1/r)} -\frac{\Big(\sum_{l=1}^{n}b_{jk}^{(l)}\Big)^2}{n+\mathcal{O}(1/r)}+\mathcal{O}(1/r)\Big]\\
			&\phantom{-}=\frac{r}{2\sigma^2}\Big[\frac{\Big(\sum_{l=1}^{t-1}b_{jk}^{(l)}\Big)^2}{(t-1)}+\frac{\Big(\sum_{l=t}^{n}b_{jk}^{(l)}\Big)^2}{(n-t+1)} -\frac{\Big(\sum_{l=1}^{n}b_{jk}^{(l)}\Big)^2}{n}\Big]+\mathcal{O}(1).
		\end{split}
	\end{align}
	Let $U=\sum_{l=1}^{t-1}b_{jk}^{(l)}$ and $V=\sum_{l=t}^{n}b_{jk}^{(l)}$. Consider a random variable $W$ which has the following distribution:
	\begin{equation}
		W=
		\begin{cases}
			\frac{U}{t-1}, & \text{with probability}\ \frac{t-1}{n}, \\
			\frac{V}{n-t+1}, & \text{with probability}\ \frac{n-t+1}{n}.
		\end{cases}
	\end{equation}
	Let $\psi(w)=w^2$. Then by Jensen's inequality, we have
	\begin{align}
		\label{eqn:eqlabel}
		\begin{split}
			&\Big[\frac{U}{t-1}\frac{t-1}{n}+\frac{V}{n-t+1}\frac{n-t+1}{n}\Big]^2\\ &\phantom{\Big[}\leq\Big[\Big(\frac{U}{t-1}\Big)^2\frac{t-1}{n}+\Big(\frac{V}{n-t+1}\Big)^2\frac{n-t+1}{n}\Big].
		\end{split}
	\end{align}
	That is
	\begin{align}
		\label{eqn:eqlabel}
		\begin{split}
			\frac{U^2}{t-1}+\frac{V^2}{n-t+1}-\frac{(U+V)^2}{n}\geq 0.
		\end{split}
	\end{align}
	Thus, the term in the brackets of the exponential in (47) is nonnegative. Hence it suffices to control its expectation and show that it remains bounded as $r\rightarrow\infty$. Suppose that the true value is $\beta_{jk,0}^{(i)}=\xi$. Then the expectation of (48) with respect to the true value is 
	\begin{align}
		\label{eqn:eqlabel}
		\begin{split}
			&\frac{r}{2\sigma^2}\Big[\frac{(t-1)^2\xi^2+\displaystyle\frac{(t-1)\sigma^2}{r}}{(t-1)}+\frac{(n-t+1)^2\xi^2+\displaystyle\frac{(n-t+1)\sigma^2}{r}}{(n-t+1)}\\
			&\phantom{\frac{r}{2\sigma^2}}-\frac{n^2\xi^2+\displaystyle\frac{n\sigma^2}{r}}{n}\Big]+\mathcal{O}(1)\\
			&\phantom{-}=\frac{r\xi^2}{2\sigma^2}\Big[(t-1)+(n-t+1)-n\Big]+\mathcal{O}(1),
		\end{split}
	\end{align}
	and the first term vanishies. Thus the exponential term in (47) is bounded in probability. Hence (47) goes to 0 as $r\rightarrow\infty$. 
	\par If the true parameter has Structure 5, and the compatible model is Structure 1, then we have the following marginal likelihood ratio: 
	\begin{align}
		\label{eqn:eqlabel}
		\begin{split}
			&\frac{\text{P}(b_{jk}^{(1)},\dots,b_{jk}^{(n)}|S_{jk}^1)}{\text{P}(b_{jk}^{(1)},\dots,b_{jk}^{(n)}|S_{jk}^5)}\\
			&=\big(c_j^2r(t-1)+1\big)^{-1/2}\big(c_j^2r(n-t+1)+1\big)^{-1/2}\\
			&\phantom{\big(c_j^2}\times\exp\Big\{\frac{r}{2\sigma^2}\Big[\frac{c_j^2}{c_j^2(t-1)+1/r}\Big(\sum_{l=1}^{t-1}b_{jk}^{(l)}\Big)^2+\frac{c_j^2}{c_j^2(n-t+1)+1/r}\Big(\sum_{l=t}^{n}b_{jk}^{(l)}\Big)^2\Big]\Big\}.
		\end{split}
	\end{align}
	The first two terms with the squre root goes to 0 as $r\rightarrow\infty$. Similarly, we have 
	\begin{align}
		\label{eqn:eqlabel}
		\begin{split}
			&\frac{r}{2\sigma^2}\left[\frac{c_j^2}{c_j^2(t-1)+1/r}\left(\sum_{l=1}^{t-1}b_{jk}^{(l)}\right)^2+\frac{c_j^2}{c_j^2(n-t+1)+1/r}\left(\sum_{l=t}^{n}b_{jk}^{(l)}\right)^2\right]\\
			&\phantom{\frac{r}{2\sigma^2}}=\frac{r}{2\sigma^2}\left[\frac{\left(\sum_{l=1}^{t-1}b_{jk}^{(l)}\right)^2}{(t-1)}+\frac{\left(\sum_{l=t}^{n}b_{jk}^{(l)}\right)^2}{(n-t+1)}\right]+\mathcal{O}(1),
		\end{split}
	\end{align}
	which is always nonnegative.
	Since the true value is $\beta_{jk,0}^{(i)}=0$, then the expectation of (54) with respect to the true value is
	\begin{align}
		\label{eqn:eqlabel}
		\begin{split}
			&\frac{r}{2\sigma^2}\left[\frac{\displaystyle\frac{(t-1)\sigma^2}{r}}{(t-1)}+\frac{\displaystyle\frac{(n-t+1)\sigma^2}{r}}{(n-t+1)}\right]+\mathcal{O}(1),
		\end{split}
	\end{align}
	which is bounded. Hence we can conclude that the exponential term in (53) is bounded in probability, and the whole expression in (53) goes to 0 as  $r\rightarrow\infty$. 
	\par Similarly, we can show that the marginal likelihood ratio also goes 0 when the compatible models are Structure 2--4. 
	\par Structure 5 is the true model, so the marginal likelihood ratio is 1 if the compatible model is structure 5.
\end{proof}

\section{Simulation}
\par In order to study the performance of our method, we implement it on a set of simulated data. Since we detect the change point through the features extracted by the DWT, we first generate the wavelet coefficients and then apply the inverse discrete wavelet transform to get the functional data. We generate 16 features for first data point by a uniform distribution on $[0, 0.5]$. To make the change distinguishable, we generate 16 features for last data point by a uniform distribution on $[0.5, 1]$. Suppose that we have 100 data points, and there is one change in each feature. We randomly sample 16 numbers from 1 to 100 and regard them as the change point for the 16 features. To generate the sequence of 100 data, we repeat the feature of the first data point and change it to that of the last data point after the change point. Thus, we have a sequence of data representated by the true feature values. After applying the inverse discrete wavelet transform, we get a sequence of 100 true signals. To generate features for the data with different noise, we sample from the normal distribution with the true feature values as the mean value and variance of 0.01, 0.1, and 1, respectively. Hence we get three sequences of 100 functional observations after applying the inverse discrete wavelet transform to them. 
\par We apply our method to the observations. For the true signals, we use (24) to compute the similarity. The change point is the value of $k$ where $C(k)$ in (25) is the minimum. Once we detect the first change point, we divide the sequence of data into two subgroups. Furthermore, we can find the change point in these two groups. We can continue the process to divide the data into more subgroups, and stop either the plot of $C(k)$ versus $k$ is relatively flat which means that there is not much difference in these data, or the minimum number of data points is reached, or the the maximum step of the resulting binary tree is reached. In this study, we stop either if  $\text{max}(C(k))-\text{min}(C(k))<0.1$, or there are less than 10 data points in the group, or the resulting binary tree has 3 steps. We compare our results with the E-Divisive method \citep{JamesR} in the R package $\tt{ecp}$, which also estimates multiple change points by iteratively applying a procedure for locating a single change point. We apply the E-Divisive method on the wavelet transform of the observations. Table 2 shows the change points for different sequences of observations detected by our method and E-Divisive method. The numbers in the parentheses denote the hierarchical order of the change points. When the variance is small (0.01), the change points our method detects are exactly the same as the true change points. With a larger variance (0.1), our method still can detect the most of the change points correctly. When the variance is large, naturally it would be difficult to detect the change points by any method.
\begin{table}
	\caption{Change point detection results comparison}
	\centering
	\begin{tabular}{ccc}
		\hline
		Method & Data & Change points \\ 
		\hline
		     & True signal& $26^{(2)}$, $49^{(1)}$, $69^{(2)}$\\
		Our& Observations with variance 0.01 & $26^{(2)}$, $49^{(1)}$, $69^{(2)}$\\
		E-Divisive& Observations with variance 0.01 & $25^{(2)}$, $49^{(1)}$, $75^{(2)}$\\
		Our& Observations with variance 0.1 & $26^{(2)}$, $49^{(1)}$, $68^{(3)}$, $75^{(2)}$, $92^{(3)}$\\
		E-Divisive& Observations with variance 0.1 & $25^{(2)}$, $49^{(1)}$, $62^{(3)}$, $75^{(2)}$, $88^{(3)}$\\
		Our& Observations with variance 1 & $5^{(3)}$, $24^{(2)}$, $40^{(3)}$, $46^{(1)}$, $94^{(3)}$, $99^{(2)}$\\
		E-Divisive& Observations with variance 1 & $42^{(1)}$, $70^{(2)}$\\
		\hline
	\end{tabular}
\end{table}

\section{Application}
\par On Berkeley Earth (\url{http://berkeleyearth.org/data/}), we can find the land-surface monthly average temperature between 1753--2016. These temperatures are 
in degrees Celsius and reported as anomalies relative to the average temperature from Jan. 1951 to Dec. 1980. We can construct a set of functional data by the 12 monthly average temperatures in each year. We smooth the data by the basis expansion. Thus we get 264 functional data ordered by the year. Figure 1 shows the plot of the 264 functional data. We believe there is a change in these functional data. Figure 2 displays the curves for every 66 years, and we can see the change in the pattern of curves.
\par Figure 3 is the plot of $C(k)$ versus different $k$. We detect the change point of this sequence of functional data at the year 1914. Figure 4 shows the curves before the change point which are the years 1753--1913 and the curves after the change point which are the years 1914--2016. We can see that the patterns are very different in these two plots. 
\par Furthermore, we can find the change point in these two subgroups. Figure 5 is the plot of $C(k)$ versus $k$ between 1753--1913, and we detect the change point at year 1839. Figure 6 shows the curves before the change point which are the years 1753--1838 and the curves after the change point which are the years 1839--1913. Figure 7 is the plot of $C(k)$ versus $k$ between 1914--2016, and we detect the change point at year 1969. Figure 8 shows the curves before the change point which are the years 1914--1968 and the curves after the change point which are the years 1969--2016. Hence we divide the data into four subgroups. We continue the process to divide the data into more subgroups, and stop if $\text{max}(C(k))-\text{min}(C(k))<0.1$. We generate 15 subgroups. Figure 9 demonstrates the hierarchical structure in the subgroups. 
\par In Section 6, we assume that we have $r$ replications of the data. To align with this assumption, we may group the data by every several consecutive years, and treat each group as a non-separable block. We can represent each block by the most representative pattern. For example, we can group the climate data by every 10 years and get 26 blocks. There are multiple ways to generate the most representative patterns, hence we can get different replications of the 26 blocks. Note that when we choose the number of observations in each block, we need to make sure that there is no distinct difference in patterns in that block.

\begin{figure*}
	\begin{subfigure}[b]{0.45\textwidth}
		\centering
		\includegraphics[width=\textwidth]{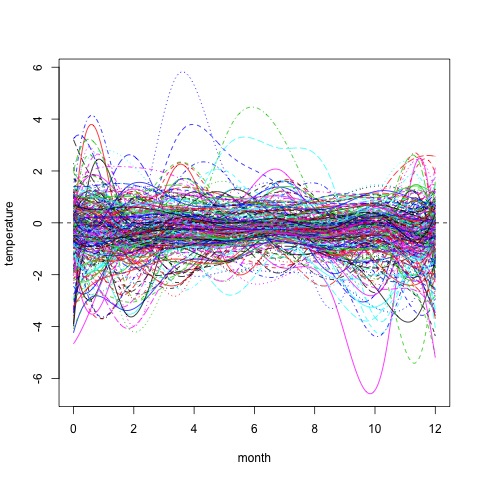}
	\end{subfigure}
	\caption{Land-surface average temperature curves between 1753--2016}
	\vskip\baselineskip
	\begin{subfigure}[b]{0.45\textwidth}
		\centering
		\includegraphics[width=\textwidth]{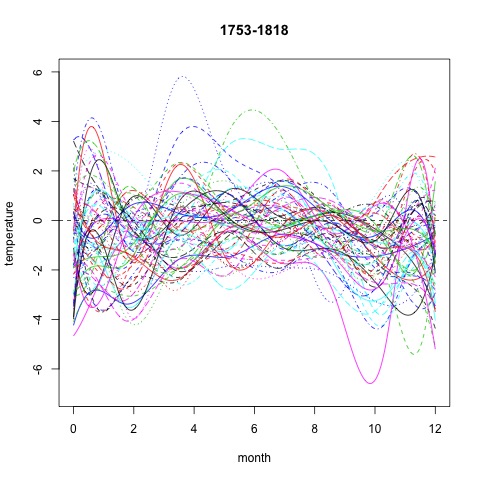}
	\end{subfigure}
	\quad
	\begin{subfigure}[b]{0.45\textwidth}
		\centering
		\includegraphics[width=\textwidth]{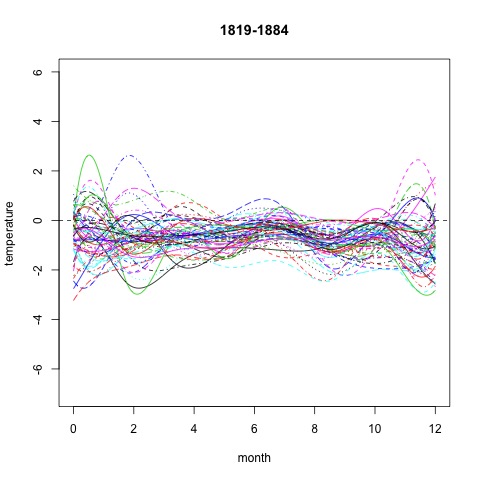}
	\end{subfigure}
	\vskip\baselineskip
	\begin{subfigure}[b]{0.45\textwidth}
		\centering
		\includegraphics[width=\textwidth]{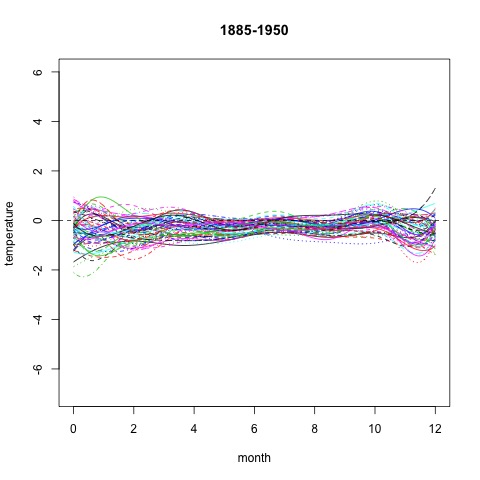}
	\end{subfigure}
    \quad
	\begin{subfigure}[b]{0.45\textwidth}
		\centering
		\includegraphics[width=\textwidth]{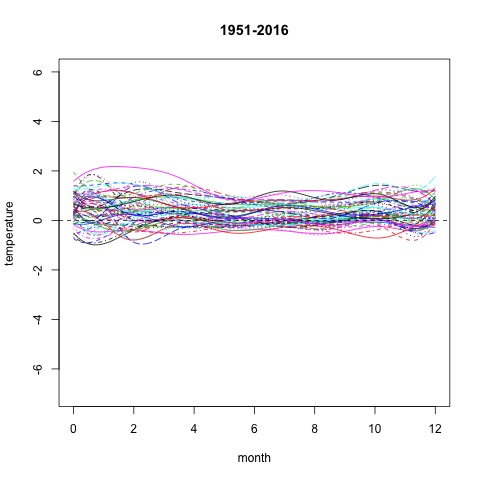}
	\end{subfigure}
	\quad
	\caption{Land-surface average temperature curves for every 66 years}
\end{figure*}

\begin{figure*}
	\begin{subfigure}[b]{0.6\textwidth}
		\centering
		\includegraphics[width=\textwidth]{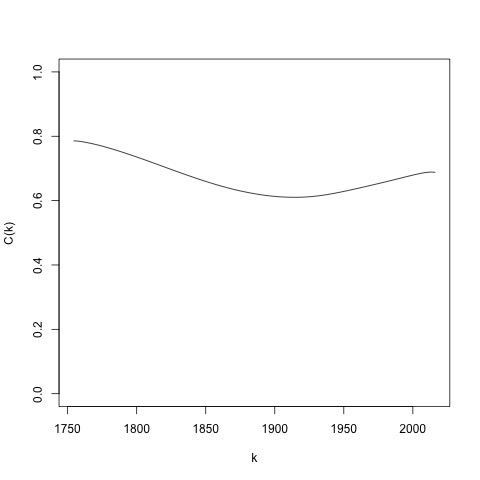}
	\end{subfigure}
    \caption{$C(k)$ for different $k$ between 1753--2016}
    \vskip\baselineskip
	\begin{subfigure}[b]{0.45\textwidth}
		\centering
		\includegraphics[width=\textwidth]{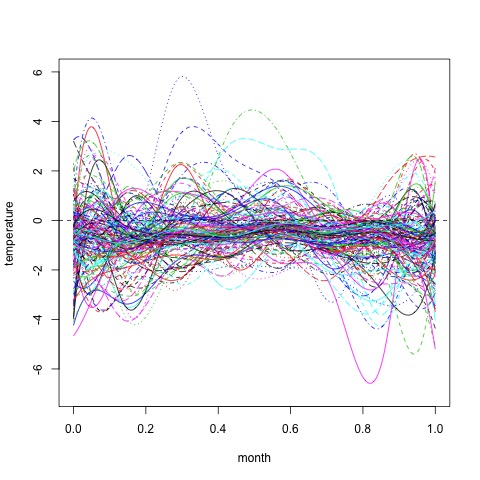}
	\end{subfigure}
	 \quad
	\begin{subfigure}[b]{0.45\textwidth}
		\centering
		\includegraphics[width=\textwidth]{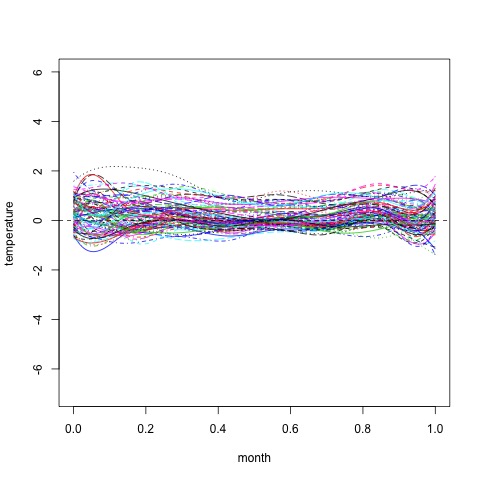}
	\end{subfigure}
	\caption{The plot on the left is the land-surface average temperature curves between 1753--1913. The plot on the right is the land-surface average temperature curves between 1914--2016.}
\end{figure*}

\begin{figure*}
	\begin{subfigure}[b]{0.6\textwidth}
		\centering
		\includegraphics[width=\textwidth]{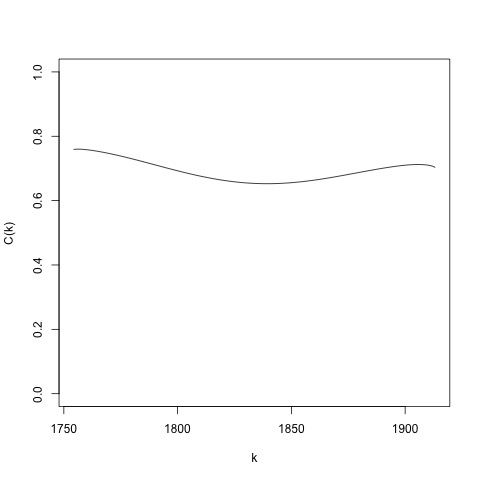}
	\end{subfigure}
	\caption{$C(k)$ for different $k$ between 1753--1913}
	\vskip\baselineskip
	\begin{subfigure}[b]{0.45\textwidth}
		\centering
		\includegraphics[width=\textwidth]{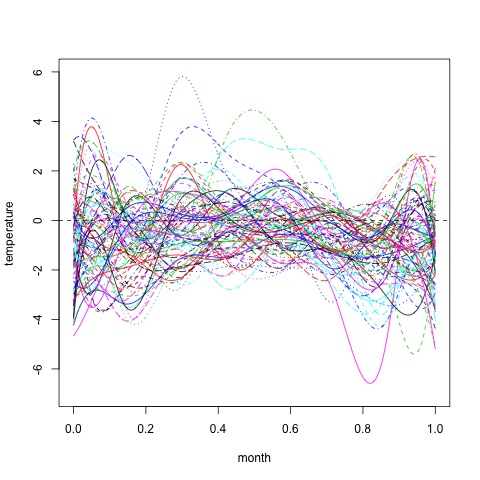}
	\end{subfigure}
	\quad
	\begin{subfigure}[b]{0.45\textwidth}
		\centering
		\includegraphics[width=\textwidth]{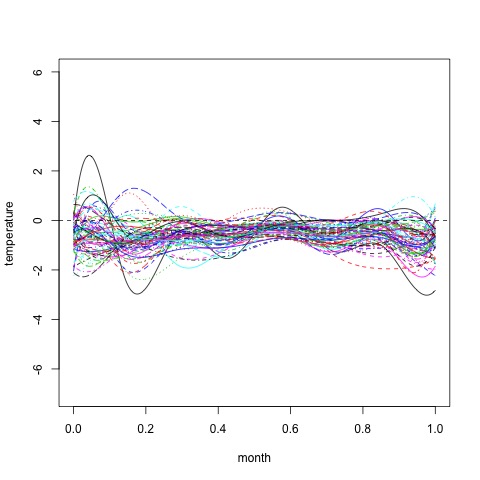}
	\end{subfigure}
	\caption{The plot on the left is the land-surface average temperature curves between 1753--1838. The plot on the right is the land-surface average temperature curves between 1839--1913.}
\end{figure*}

\begin{figure*}
	\begin{subfigure}[b]{0.6\textwidth}
		\centering
		\includegraphics[width=\textwidth]{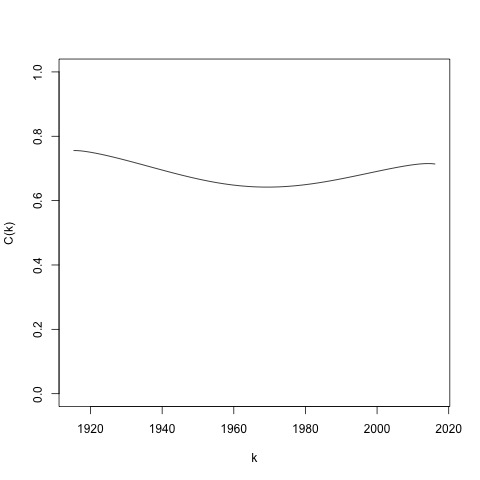}
	\end{subfigure}
	\caption{$C(k)$ for different $k$ between 1914--2016}
	\vskip\baselineskip
	\begin{subfigure}[b]{0.45\textwidth}
		\centering
		\includegraphics[width=\textwidth]{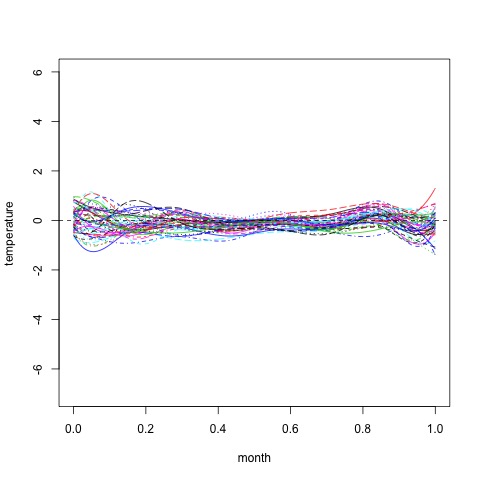}
	\end{subfigure}
	\quad
	\begin{subfigure}[b]{0.45\textwidth}
		\centering
		\includegraphics[width=\textwidth]{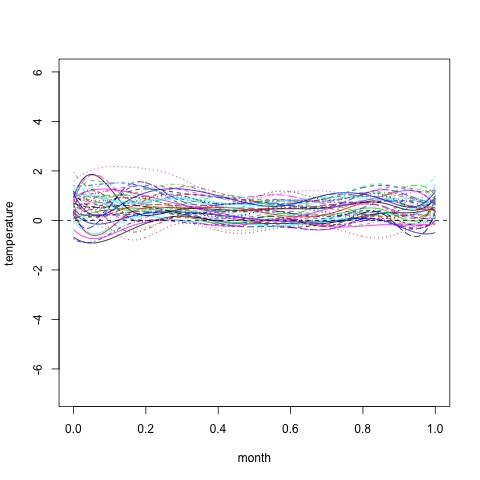}
	\end{subfigure}
	\caption{The plot on the left is the land-surface average temperature curves between 1914--1968. The plot on the right is the land-surface average temperature curves between 1969--2016.}
\end{figure*}

\begin{figure*} 
\begin{forest}
 for tree={grow'=0,
 l sep=0.1cm,
 s sep=1cm},
 [1753-2016 [1753-1913[1753-1838[1753-1800[1753-1778[1753-1766[1753-1760][1761-1766]][1767-1778]][1779-1800]][1801-1838[1801-1820[1801-1811][1812-1820]][1821-1838]]][1839-1913[1839-1877[1839-1857][1858-1877]][1878-1913[1878-1897][1898-1913]]]] [1914-2016[1914-1968[1914-1941][1942-1968]][1969-2016[1969-1993][1994-2016]]] ]
\end{forest}
\caption{Subgroups between 1753--2016}
\end{figure*}
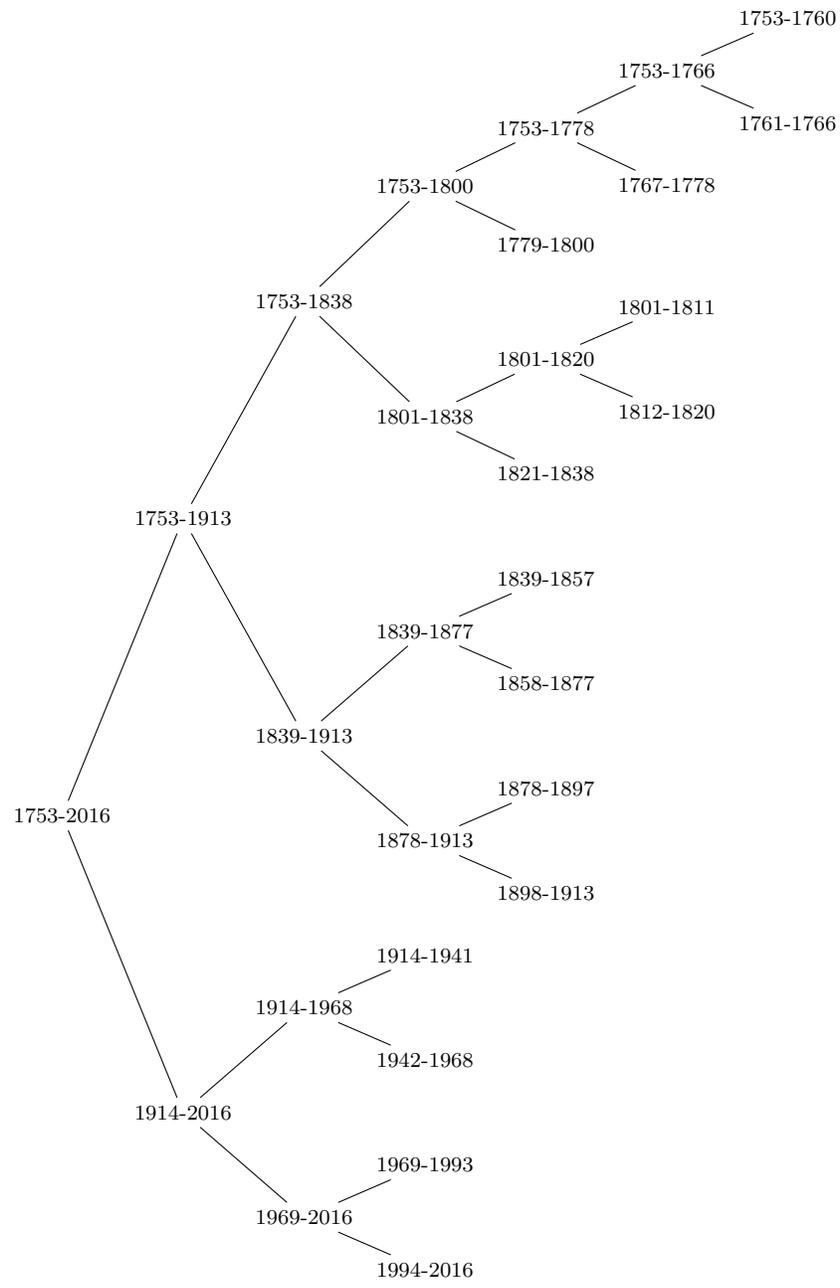

\bibliographystyle{ba}
\bibliography{ref}

\begin{thebibliography}{12}
\newcommand{\enquote}[1]{``#1''}
\expandafter\ifx\csname natexlab\endcsname\relax\def\natexlab#1{#1}\fi
\expandafter\ifx\csname url\endcsname\relax
  \def\url#1{{\tt #1}}\fi
\expandafter\ifx\csname urlprefix\endcsname\relax\def\urlprefix{URL }\fi
\ifx\endbibitem\undefined \let\endbibitem\relax\fi

\bibitem[{Abramovich et~al.(1998)Abramovich, Sapatinas, and
  Silverman}]{Abramovich1998}
Abramovich, F., Sapatinas, T., and Silverman, B.~W. (1998).
\newblock \enquote{Wavelet thresholding via a Bayesian approach.}
\newblock {\em Journal of the Royal Statistical Society. Series B\/}, 60:
  725--749.
\endbibitem

\bibitem[{Antoniadis et~al.(2013)Antoniadis, Brossat, Cugliari, and
  Poggi}]{Antoniadis}
Antoniadis, A., Brossat, X., Cugliari, J., and Poggi, J. (2013).
\newblock \enquote{Clustering functional data using wavelets.}
\newblock {\em International Journal of Wavelets, Multiresolution and
  Information Processing\/}, 11: 1350003--1350032.
\endbibitem

\bibitem[{Aston and Kirch(2012)}]{Aston}
Aston, J. and Kirch, C. (2012).
\newblock \enquote{Detecting and estimating changes in dependent functional
  data.}
\newblock {\em Journal of Multivariate Analysis\/}, 109: 204--220.
\endbibitem

\bibitem[{Aue et~al.(2018)Aue, Rice, and S\"{o}nmez}]{Aue2018}
Aue, A., Rice, G., and S\"{o}nmez, O. (2018).
\newblock \enquote{Detecting and dating structural breaks in functional data
  without dimension reduction.}
\newblock {\em Journal of the Royal Statistical Society. Series B\/}, 80:
  509--529.
\endbibitem

\bibitem[{Berkes et~al.(2009)Berkes, Gabrys, Horv\'{a}th, and
  Kokoszka}]{Berkes}
Berkes, I., Gabrys, R., Horv\'{a}th, L., and Kokoszka, P. (2009).
\newblock \enquote{Detecting changes in the mean of functional observations.}
\newblock {\em Journal of the Royal Statistical Society. Series B (Statistical
  Methodology)\/}, 71: 927--946.
\endbibitem

\bibitem[{Ghosal and \noopsort{Vaart}{van der Vaart}(2017)}]{Ghosal}
Ghosal, S. and \noopsort{Vaart}{van der Vaart}, A. (2017).
\newblock {\em Fundamentals of nonparametric Bayesian inference\/}.
\newblock Cambridge, UK: Cambridge University Press.
\endbibitem

\bibitem[{James and Matteson(2014)}]{JamesR}
James, N.~A. and Matteson, D.~S. (2014).
\newblock \enquote{ecp: An R package for nonparametric multiple change point
  analysis of multivariate data.}
\newblock {\em Journal of Statistical Software\/}, 62(7).
\endbibitem

\bibitem[{Lian(2011)}]{Lian}
Lian, H. (2011).
\newblock \enquote{On posterior distribution of Bayesian wavelet thresholding.}
\newblock {\em Journal of Statistical Planning and Inference\/}, 141: 318--324.
\endbibitem

\bibitem[{Schwartz(1965)}]{Schwartz}
Schwartz, L. (1965).
\newblock \enquote{On Bayes procedures.}
\newblock {\em Probability Theory and Related Fields\/}, 4: 10--26.
\endbibitem

\bibitem[{Sharipov et~al.(2016)Sharipov, Tewes, and Wendler}]{Sharipov}
Sharipov, O., Tewes, J., and Wendler, M. (2016).
\newblock \enquote{Sequential block bootstrap in a Hilbert space with
  application to change point analysis.}
\newblock {\em The Canadian Journal of Statistics\/}, 44: 300--322.
\endbibitem

\bibitem[{Suarez and Ghosal(2016)}]{Suarez}
Suarez, A. and Ghosal, S. (2016).
\newblock \enquote{Bayesian clustering of functional data using local
  features.}
\newblock {\em Bayesian Analysis\/}, 11: 71--98.
\endbibitem

\bibitem[{Zhang et~al.(2011)Zhang, Shao, Hayhoe, and Wuebbles}]{Zhang}
Zhang, X., Shao, X., Hayhoe, K., and Wuebbles, D. (2011).
\newblock \enquote{Testing the structural stability of temporally dependent
  functional observations and application to climate projections.}
\newblock {\em Electronic Journal of Statistics\/}, 5: 1765--1796.
\endbibitem

\end{thebibliography}

\end{document}